\documentclass{llncs}

\usepackage{amsmath,amssymb,array,color,colortbl,graphicx,multirow}
\usepackage{microtype}
\usepackage[pdfpagelabels=true]{hyperref}
\usepackage{comment}

\usepackage{wrapfig}
\usepackage{algorithm}
\usepackage[noend]{algorithmic}

\newcommand{\Rob}{\textsc{Rob}}
\newcommand{\rfs}{\mathbf{RFS}}
\newcommand{\dfs}{\mathbf{DFS}}

\newcommand{\Bal}{\textsc{Bal}}

\newcommand{\loadmax}{\widehat{\lambda}}

\def\E{\mathbb{E}}

\newtheorem{clm}[theorem]{Claim}

\begin{document}

\title{How (Not) to Shoot in Your Foot\\with SDN Local Fast Failover\\{\large A Load-Connectivity Tradeoff}}

\author{
Michael Borokhovich\inst{1,}\thanks{Michael Borokhovich was supported in part by the Israel Science Foundation (grant 894/09).}, Stefan Schmid\inst{2}
}


\institute{Ben-Gurion University of the Negev, Israel,  \email{borokhom@cse.bgu.ac.il}
\and
TU Berlin \& T-Labs, Germany,
\email{stefan@net.t-labs.tu-berlin.de}
}



\maketitle


\begin{abstract}
{
This paper studies the resilient routing and (in-band) fast failover mechanisms supported in Software-Defined Networks (SDN).
We analyze the potential benefits and limitations of such failover mechanisms, and focus on two main metrics: (1) \emph{correctness} (in terms of connectivity and loop-freeness) and (2) \emph{load-balancing}. We make the following contributions.
First, we show that in the \emph{worst-case} (i.e., under adversarial link failures), the usefulness of local failover is rather limited: already a small number of failures will violate connectivity properties under \emph{any} fast failover policy, even though the underlying substrate network remains highly connected. We then present randomized and deterministic algorithms to compute resilient forwarding sets; these algorithms achieve an almost optimal tradeoff.
 Our worst-case analysis is complemented with a simulation study.
}
\end{abstract}

\section{Introduction}

The \emph{software-defined networking (SDN)} paradigm separates the control plane from the network data plane, and introduces a (software) \emph{controller} that manages the \emph{flows} in the network from a (logically) centralized perspective. This architecture has the potential to make the network management and operation more flexible and simpler, and to enable faster innovation also in the network core. For example, the controller may exploit application and network state information (including the switches under its control) to optimize the routing of the flows through the network, e.g., to implement isolation properties or improve performance.

However, the separation of the control from the data plane may have drawbacks. For example, a reactive flow control can introduce higher latencies due to the interaction of the switch with the remote controller. Moreover, the separation raises the question of what happens if the switches lose connectivity to the controller. One solution to mitigate these problems is to keep certain functionality closer to the switches or in the data plane~\cite{kandoo}.

An important tradeoff occurs in the context of network failures: Theoretically, e.g., a link failure, is best handled by the controller which has the logic to update forwarding rules according to the current network policies. However, as the indirection via the controller may take too long, modern network designs incorporate failover (or ``backup'') paths into the (switches' or routers') \emph{forwarding tables}. For example, OpenFlow (since the 1.1 specification~\cite{oneone}), incorporates such a fast failover mechanism: it allows to predefine resilient and in-band failover routes which kick in upon a topological change. Only after the failover took place, the controller may learn about the new situation and install forwarding (and failover) rules accordingly.

\textbf{Our Contributions.} Given that the in-band failover tables need to be pre-computed and the corresponding rules are based on limited local network information only, we ask the question: ``Can you shoot in your foot with local fast failover?''
We formalize a simplified local failover problem, and first assume a conservative (or worst-case) perspective where
link failures are chosen by an adversary who knows the entire network and all the pre-installed failover rules.
For this setting, we provide a lower bound which shows that a safe fast failover can potentially come at a high network load, especially
if the failover rules are destination-based only (Section~\ref{sec:worstcase}).
We then present randomized and deterministic algorithms to pre-compute resilient forwarding sets and show that the algorithms are (almost) optimal in the sense that they match the
 lower bound mentioned above (Section~\ref{sec:algorithms}).
 Finally, we report on a simulation study (Section~\ref{sec:sims}) which indicates that under random link failures,
 local fast failover performs better in general.
In the Appendix, we give the formal specification of the two additional algorithms used in our simulations, and we extend the discussion to alternative adversary and traffic models.



\textbf{Model and Terminology.}
We attend to the following model. We assume an SDN-network $G=(V,E)$ with $n$ switches (or \emph{nodes}) $V=\{v_1,\ldots,v_n\}$ (e.g., OpenFlow switches) connected by bidirectional links $E$. We assume that all nodes are directly connected, i.e., $G$ forms a full mesh (a \emph{clique}). This network serves an \emph{all-to-one} communication pattern where any node $v_i\in V\setminus \{v_n\}$ communicates with a single destination $v_n$; in other words, we have $n-1$ communicating (source-destination) pairs. Henceforth, by slightly abusing terminology, we will refer to the corresponding $n-1$ communication paths as the \emph{flows} $\mathcal{F}=\{f_1,\ldots,f_{n-1}\}$. The source-destination flows are unsplittable, i.e., each flow $f_i~~\forall i\in \{1,\ldots,k\}$ travels along a single path. For simplicity, we will assume that all flows $f_i$ carry a constant amount of traffic $w=w(f_i)$, and that edge capacities $e\in E$ are infinite.

In order to ensure an efficient failover, each switch $v\in V$ can store the following kind of \emph{failover rules}: Each rule $r\in R$ considers a specific local failure scenario, namely the set of failed incident links, and defines an alternative forwarding port for each source-destination pair. (This is slightly more general than what is provided e.g., by OpenFlow today: in OpenFlow, all paths need to resort to the same failover port, rending the connectivity-load tradeoff even worse.)

Formally, let $\Gamma(v)~~ \forall v\in V$ denote the links (or equivalently: the \emph{switch ports}) incident to node $v$ in $G$, and let $\texttt{FW}(v)$ define how the source-destination pairs (or flows) that are routed via node $v$ (the ``forwarding set''). A rule $r$ is of the form:
$
r: ~~ \left(2^{\Gamma(v)}, \texttt{FW}(v) \right) \mapsto  \texttt{FW}(v),
$
that is, for each possible failure scenario $2^{\Gamma(v)}$ (i.e., the subset of ports which failed at $v$), the failover rule defines an alternative set of forwarding rules $\texttt{FW}(v)$ at $v$. Note that the number of rules can theoretically be large; however, as we will see, small rule tables are sufficient for the algorithms presented in this paper.


We study failover schemes that pursue two goals: (1) \emph{Correctness:} Each source-destination pair is connected by a valid path; there are no forwarding loops.
(2) \emph{Performance:} The resulting flow allocations are well balanced. Formally, we want to minimize the load of the maximally loaded link in $G$ after the failover:
    $
\min \max_{e\in E} \lambda(e)$,
 where $\lambda(e)$ describes the number of flows $f_i$ crossing edge $e$. Henceforth, let $\widehat{\lambda}=\max_{e\in E} \lambda(e)$ denote the maximum load.

For our randomized failover schemes, we will typically state our results \emph{with high probability} (short: \emph{w.h.p.}): this means that the corresponding claim holds with at least polynomial probability $1-1/n^c$ for an arbitrary constant $c$. Moreover, throughout this paper, $\log$ will refer to the \emph{binary} logarithm.

\section{You must shoot in your foot!}\label{sec:worstcase}

Let us first investigate the limitations of local failover mechanisms from a conservative worst-case perspective. Concretely, we will show that even in a fully meshed network (i.e., a \emph{clique}), a small number of link failures can either quickly disrupt connectivity (i.e., the forwarding path of at least one source-destination pair is incorrect), or entail a high load. This is true even though the remaining physical network is still well connected: the minimum edge cut (short: \emph{mincut}) is high, and there still exist many disjoint paths connecting each source-destination pair.
\begin{theorem}\label{thm:worstcase}
No local failover scheme can tolerate $n-1$ or more link failures without disconnecting source-destination pairs, even though the remaining graph (i.e., after the link failures) is still $\left\lfloor \tfrac{n}{2}\right\rfloor -1$-connected.
\end{theorem}
\begin{proof}
We consider a physical network that is fully meshed, and we assume a traffic matrix where all nodes communicate with a single destination $v_n$.
To prove our claim, we will construct a set of links failures that creates a loop, for any local failover scheme. Consider a flow $v_1 \rightarrow v_n$ connecting
the source-destination pair $(v_1,v_n)$. The idea is that whenever the flow from $v_1$ would be directly forwarded to $v_n$ in the absence of failures,
 we fail the corresponding physical link: that is, if $v_1$ would directly forward to $v_n$, we fail $(v_1,v_n)$. Similarly, if $v_1$ forwards to (the backup) node $v_i$, and if $v_i$ would send to $v_n$, we fail $(v_i,v_n)$, etc.
We do so until the number of intermediate (backup) nodes for the flow $v_1\rightarrow v_n$ becomes $\left\lfloor \frac{n}{2}\right\rfloor -1$. This will require at most $\left\lfloor \frac{n}{2}\right\rfloor -1$ failures (of links to $v_n$) since every such failure adds at least one intermediate node.

In the following, let us assume that the last link on the path $v_1\rightarrow v_n$ is $(v_k,v_n)$.
We simultaneously fail all the links $(v_k,v_{*})$, where $v_{*}$ are all the nodes that are not the intermediate nodes on the path $v_1\rightarrow v_n$, and not $v_1$.
So, there are $n-\left(\left\lfloor \frac{n}{2}\right\rfloor - 1\right) -1$ nodes $v_{*}$ (the last minus $1$ accounts for $v_1$).
By failing the links to $v_{*}$, we left $v_k$ without a valid routing choice: All the remaining links from $v_k$ point to nodes which are already on the path $v_1\rightarrow v_n$, and a loop is inevitable.


In total, we have at most $\left\lfloor \frac{n}{2}\right\rfloor -1 + n-\left(\left\lfloor \frac{n}{2}\right\rfloor - 1\right) -1=n-1$ failures.
Notice, that the two nodes with the smallest degrees in the graph are the nodes $v_n$ and $v_k$. The latter is true since the first $\left\lfloor \frac{n}{2}\right\rfloor -1$ failures were used to disconnect links to $v_n$, and another $n-\left\lfloor \frac{n}{2}\right\rfloor$ failures were used to disconnect links from $v_k$.
Formally, $d(v_n)=(n-1)-\left(\left\lfloor \frac{n}{2}\right\rfloor -1\right)-1$, where the last minus $1$ accounts for the link $(v_k,v_n)$. So, $d(v_n)=n-1-\left\lfloor \frac{n}{2}\right\rfloor \ge \tfrac{n}{2}-1$. And regarding the degree of $v_k$, we have: $d(v_k)=(n-1)-\left(n- \left\lfloor \frac{n}{2}\right\rfloor\right)=\left\lfloor \frac{n}{2}\right\rfloor -1$.
All the other nodes have a degree of $n-2$.

The network is still $\left\lfloor \frac{n}{2}\right\rfloor - 1$ connected: the mincut of the network is at least $\left\lfloor \frac{n}{2}\right\rfloor - 1$. Consider some cut with $k$ nodes on the one side of the cut, and $n-k$ nodes on the other side. Obviously, one of the sets has a size of at most $\left\lfloor\frac{n}{2}\right\rfloor$; let us denote this smaller set by $S$. If $S$ includes at least one of the nodes $V\setminus \{v_k,v_n\}$, then the number of outgoing edges form the set is at least $n-2-(|S|-1)$, thus the mincut is at least $\frac{n}{2} - 1$. If $S$ includes only both $v_k$ and $v_n$, the mincut is at least $n-1$ (the link $(v_k,v_n)$ was failed). If only one of the nodes $\{v_k,v_n\}$ is in $S$, then the mincut is at least $\left\lfloor \frac{n}{2}\right\rfloor - 1$.
\hfill $\Box$\end{proof}

Regarding the maximal link load, we have the following lower bound.
\begin{theorem}\label{thm:lowerbound}
For any local failover scheme tolerating $\varphi$ link failures $(0<\varphi<n)$ without disconnecting any source-destination pair, there exists a failure scenario which results in a link load of at least $\loadmax\ge\sqrt{\varphi}$, although the minimum edge cut (mincut) of the network is still at least $n-\varphi-1$.
\end{theorem}
\begin{proof}
\emph{Let us first describe an adversarial strategy that induces a high load:} Recall that in the absence of failures, each node $v_i$ ($i\neq n$) may use its direct link to $v_n$ for forwarding. However,
 after some links failed, $v_i$ may need to resort to the remaining (longer) paths from $v_i$ to $v_n$.
Since the failover scheme $\mathcal{S}$ tolerates $\varphi$ failures and $v_i$ remains connected to $v_n$,
$\mathcal{S}$ will fail over to one of $\varphi+1$ possible paths. To see this, let $v_i^j$ ($j\in[1,\ldots ,\varphi]$) be one of the $\varphi$ possible last hops on the path $(v_i\rightarrow\cdots\rightarrow v_i^j \rightarrow v_n)$, and let us consider the paths generated by $\mathcal{S}$:
\begin{align*}
&(v_i\rightarrow v_n),\\
&(v_i\rightarrow\cdots\rightarrow v_i^1 \rightarrow v_n),\\
&(v_i\rightarrow\cdots\rightarrow v_i^1 \rightarrow\cdots\rightarrow v_i^2\rightarrow v_n),\\
&\ldots\\
&(v_i\rightarrow \cdots\rightarrow v_i^1\rightarrow\cdots\rightarrow v_i^2\rightarrow\cdots\rightarrow v_i^\varphi\rightarrow v_n).
\end{align*}
For example, the path $(v_i\rightarrow\cdots\rightarrow v_i^1 \rightarrow v_n)$ will be generated if the first failure is link $(v_i,v_n)$, and the path $(v_i\rightarrow\cdots\rightarrow v_i^1 \rightarrow\cdots\rightarrow v_i^2\rightarrow v_n)$ if the second failure is link $(v_i^1,v_n)$ (see Fig.~\ref{fig:failover_example} for an illustration). Notice that the last hop $v_i^j$ is unique for every path; otherwise, the loop-freeness property would be violated.

\begin{figure}[t]
\centering
\includegraphics[width=.24\columnwidth]{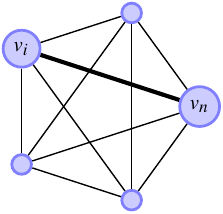}
\includegraphics[width=.24\columnwidth]{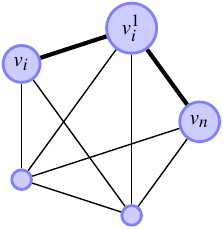}
\includegraphics[width=.24\columnwidth]{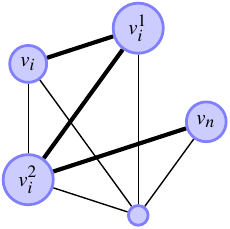}
\includegraphics[width=.24\columnwidth]{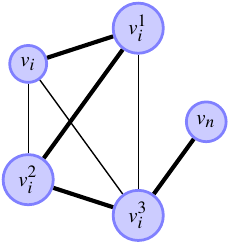}\\
\caption{\emph{From left to right:} failover path $(v_i\rightarrow v_n)$ where each time the last hop to $v_n$ is failed.}\label{fig:failover_example}
\end{figure}

For each $i\in [1,\ldots,n-1]$ (i.e., for each possible source) consider the set $A_i=\{v_i,v_i^1,\ldots,v_i^{\sqrt{\varphi}}\}$, and accordingly, the multiset $\bigcup_i A_i$ is of size $\left|\bigcup_i A_i\right| = (n-1)(\sqrt{\varphi}+1)$ many nodes. Since we have $n-1$ distinct nodes (we do not count $v_n$), by a counting argument, there exists a node $w\in \bigcup_i A_i$ which appears in at least $\sqrt{\varphi}$ sets $A_i$.

If for each $i$ such that $w\in A_i$, the adversary will cause $v_i$ to route to $v_n$ via $w$, then the load of the link $(w,v_n)$ will be at least $\sqrt{\varphi}$. This can be achieved by failing at most $\sqrt{\varphi}$ links to $v_n$ in each such set $A_i$. Thus, the adversary will fail $\sqrt{\varphi}\times \sqrt{\varphi}=\varphi$ links incident to $v_n$, while the maximum loaded link $(w,v_n)$ will have a load of at least $\sqrt{\varphi}$.

\emph{It remains to prove that the network remains highly connected, despite these failures:} The proof is simple. In a clique network without failures, the mincut is $n-1$. In the worst case, each link failure will remove one link from some cut, and hence the mincut must eventually be at least $n-\varphi-1$. By the same argument, there are at least $n-\varphi-1$ many disjoint paths from each node $v_i$ to the destination: initially, without failures, there are $n-1$ disjoint paths (a direct one and $n-1$ indirect ones), and each failure affects at most one path.
\hfill $\Box$\end{proof}

Interestingly, it can be proved analogously that if a failover rule only depends on destination addresses, the situation is even worse.
\begin{theorem}\label{thm:dst-based}
Consider any local destination-based failover scheme in a clique graph. There exists a set of $\varphi$ failures $(0<\varphi<n)$, such that the remaining graph will have a mincut of $n-\varphi-1$ and
$\loadmax\ge \varphi.$
\end{theorem}
\begin{proof}
In order to construct a bad example, we first fail the direct link $(v_1,v_n)$, and hence $v_1$ will need to reroute to some path with the last node before $v_n$ being some node $v_i$.
When we fail the link $(v_i,v_n)$, $v_i$ will have to reroute and some other node $v_j$ will become the last hop on the path to $v_n$. We repeat this strategy
to fail the links from the newly selected last hop and the destination $v_n$. This results in a routing path $v_1\rightarrow\cdots\rightarrow v_i\rightarrow\cdots\rightarrow v_j \rightarrow\cdots\rightarrow w\rightarrow v_n$ with at least $\varphi$ intermediate nodes.
Since the algorithm is destination-based, i.e., forwarding rules depend only on the destination address of a packet, the load on the link $(w,v_n)$ will be at least $\varphi+1$: all the nodes on the path $v_1\rightarrow v_n$ will send their packets via the same route.
\hfill $\Box$\end{proof}

\section{How not to shoot in your foot!}\label{sec:algorithms}

 We have seen that what can be achieved with local fast failover is rather limited. On the positive side, this section shows that there exist algorithms
 to pre-compute failover schemes which at least match the derived lower bounds: we present algorithms to pre-compute robust failover paths that \emph{jointly optimize} the loop-freeness property and the load, i.e., find an almost optimal tradeoff.

 Naturally, randomization can help to spread the communication load well, but we must ensure that paths remain loop-free. We first present such a randomized solution, discuss how to derandomize it, and finally look at deterministic failover algorithms.

We introduce a family of failover schemes $\mathcal{S}$ which can be represented in a generic \emph{matrix form} $\delta_{i,j}$.
Any failover scheme instance in this family will always forward a message directly to the destination if the corresponding link is available.
Otherwise, if a given node $v_i$ cannot reach the destination $v_n$ via $(v_i,v_n)$, it will resort to the sequence of alternatives
represented as the row $i$ in the matrix $\delta_{i,\cdot}$ (the ``backup nodes'' for $v_i$): $v_i$ will first try to forward to node $\delta_{i,1}$, if this link is not available
to node $\delta_{i,2}$, and so on. Similarly and more generally, starting from node $\delta_{i,j}$, if the link $(\delta_{i,j},v_n)$ is not
available, the failover scheme will try $\delta_{i,j+1}$, $\delta_{i,j+2}$, etc.
In summary, the matrix representation can be depicted as follows:
\begin{align*}
\delta_{1,1},\delta_{1,2},\ldots ,\delta_{1,n-2}\\
\ldots\\
\delta_{i,1},\delta_{i,2},\ldots ,\delta_{i,n-2}\\
\ldots\\
\delta_{n-1,1},\delta_{n-1,2},\ldots ,\delta_{n-1,n-2}\\
\end{align*}

The following auxiliary claim characterizes the best adversarial strategy against the failover schemes $\mathcal{S}$.
\begin{clm}\label{claim:best_adv_strategy}
For the family of failover schemes $\mathcal{S}$, the highest load is induced if links towards the destination node $v_n$ are failed.
\end{clm}
\begin{proof}
To achieve a load of $\varphi$ on some link, the adversary first needs to bring at least $\varphi$ flows to some node $w$. Consider a failover sequence $\delta_{i,\cdot}$ in which $w$ is located at $j$'s position, i.e., $\delta_{i,j}=w$. In order to bring the flow $v_i \rightarrow v_n$ to node $w$, the adversary needs to fail at least $j$ links (every failure requires at most a single additional backup node). Thus, the adversary can remove the links to the destination from every node $\delta_{i,k}, k<j$ and from the source $v_i$. The optimality is due to the fact that once one of the \emph{nodes} $\delta_{i,k}, k<j$ appears in other sequences, these failures are automatically reused: the links $(\delta_{i,k},v_n)$ already failed. If the adversary would instead choose to fail other \emph{links} (not towards the destination), e.g., $(\delta_{i,j},\delta_{i,j+1})$ , the failures can only be reused if the same \emph{link} (and not only an endpoint) appears in other sequences before $w$. Therefore, we conclude that the strategy of failing the links to the destination is optimal: (1) it requires no more failures to bring a specific flow to $w$ than any other strategy, and (2) link failures to the destination can strictly be reused more often than the failures of links to any other nodes.
\hfill $\Box$\end{proof}

\subsection{Randomized Failover}

What does a good failover matrix $\delta_{i,j}$ look like? Naively, one may choose the matrix entries (i.e., the ``failover ports'') uniformly at random from the set of next hops which are still available, and depending on the
source and destination address, in order to balance the load.
However, note that a random and \emph{independent} choice will quickly introduce loops in the forwarding sequences: it is likely that a switch will forward traffic to a switch which was already visited before on the failover path.

Thus, our randomized failover scheme $\rfs$ will choose random \emph{permutations}, i.e.,
for a source-destination pair $(v_i,v_n)$, the sequence $\delta_{i,1},\delta_{i,2},\ldots ,\delta_{i,n-2}$ (with $\delta_{i,j}\in V\setminus\{v_i,v_n\}$) is \emph{always loop-free} (deterministically).
Technically, $\rfs$ draws all $\delta_{i,j}$ uniformly at randomly from $V\setminus\{v_i,v_n\}$ but eliminates repetitions (e.g., by redrawing a repeated node).
We can show that $\rfs$ is almost optimal, in the following sense.
\begin{theorem}\label{thm:rand_strong_adv}
Using the $\rfs$ scheme, in order to create a maximum load of $\loadmax = \sqrt{\varphi}$, the adversary will have to fail at least $\Omega \left(\frac{\varphi}{\log n}\right)$ links \emph{w.h.p.}, where $0<\varphi<n$.
\end{theorem}
\begin{proof}
To create a link load of $\sqrt{\varphi}$ with the minimal number of link failures, the adversary must in particular be able to route at least $\sqrt{\varphi}$ flows
to some node $w$. Given the $\sqrt{\varphi}$ load on the node, in the best case (for the adversary), the entire flow will be forwarded by $w$ on a single outgoing link.
(E.g., the link to the destination $v_n$.)
We will show that w.h.p., it is impossible for the adversary to route more than $\sqrt{\varphi}$ flows to a single node.

The adversary can put a high load on some node $w$ only if: 1) Node $w$ is located close to the beginning of many sequences (i.e., is in a small ``prefix'' of the sequences); thus, a small number of failures is sufficient to redirect the flow to $w$. 2) Many nodes appearing before $w$ in the sequence prefixes occur early in many other prefixes as well; thus, the adversary can ``reuse'' failed links to redirect also other source-destination pairs.
Note that these two requirements may conflict, but to prove the lower bound on the number of required failures, we can assume that both conditions are satisfied: the set of $\sqrt{\varphi}$ sequences with the largest number of node repetitions in the $w$-prefixes also have the shortest $w$-prefixes.

With this intuition in mind, let us compute the probability that a node $w$ appears more than approximately $\log n$ times at position $j$. Let $Y_i^j$ be an indicator random variable that indicates whether $w$ is located at position $j\in[1,\ldots ,n-2]$ in sequence $i\in[1,\ldots ,n-1]$. Let $Y^j=\sum_{i=1}^n Y_i^j$ be a random variable representing the number of times that $w$ appears at position $j$. Since the failover sequences are random, $\Pr(Y_i^j=1)=\frac{1}{n-2}$ ($w$ is neither the source nor the destination) and thus, $\forall j, \E\left[Y^j\right] = \tfrac{n-1}{n-2}$. Applying the Chernoff bound on the sum of $n$ \emph{i.i.d.}~Poisson trials, we obtain (for any $\delta > 0$):
\begin{align}
\Pr \left(Y^j >(1+\delta)\E\left[Y^j\right]\right)& \le 2^{-\delta \E \left[Y^j\right]} \nonumber\\
\Pr \left(Y^j > \frac{(1+3\log n)(n-1)}{n-2}\right) & \le 2^{-(3\log n) \times \frac{n-1}{n-2}}\nonumber\\
&\le 2^{-3\log n}=1/n^3. \nonumber
\end{align}
Let us denote $z = \frac{(1+3\log n)(n-1)}{n-2}$ and rewrite:
\begin{align}
\Pr \left(Y^j > z\right) \le 1/n^3. \nonumber
\end{align}
We can now apply a union bound argument\footnote{The union bound argument says that the probability of the union of the events is no greater than the sum of the probabilities of the individual events.} over all possible nodes $w$ and over all possible positions $j$, which yields that with probability at least $1-\frac{1}{n}$, any node will appear no more than $z$ times at each position.


The adversary needs to select the $\sqrt{\varphi}$ sequences with the shortest $w$-prefixes. For a chosen sequence $i$, let us denote by $k_i$ the prefix length for node $w$ (the prefix length includes $w$ itself). Since each node will appear no more than $z$ times at each position (with probability of at least $1-\frac{1}{n}$) the minimum length of a \emph{total prefix} for any node $w$ can be derived. Let us denote the minimum \emph{total prefix} by $k$. Clearly, $k$ is minimized for the shortest possible prefixes $k_i$. According to the analysis above, with high probability, there are no more than $z$ prefixes of length $1$, no more than $z$ prefixes of length $2$, and so on.
Therefore:
\begin{align}
k &= \sum_{i=1}^{\sqrt{\varphi}}k_i \ge \sum_{i=1}^{z}1 + \sum_{i=1}^{z}2 +\cdots +  \sum_{i=1}^{z}\frac{\sqrt{\varphi}}{z} \nonumber \\
&= z\left(1+2+\cdots + \frac{\sqrt{\varphi}}{z}\right) =\frac{\varphi + \sqrt{\varphi}z}{2z} \ge \frac{\varphi}{2z}\nonumber \\
&\ge \frac{\varphi}{8\log n}.\label{eq:simple_den}
\end{align}
Eq.~\ref{eq:simple_den} is true since for $n\ge 6$, $\frac{(1+3\log n)(n-1)}{n-2} \le 8\log n$.

In conclusion, we know that in order to achieve a load of $\sqrt{\varphi}$, the adversary has to fail the entire total prefix of $w$ that consists of at least $\frac{\varphi}{8\log n}$ nodes.
However, the nodes in the prefixes are not necessarily all distinct, and the number of links the adversary needs to fail only depends on the \emph{distinct} nodes in the \emph{total prefix} of the node $w$. The latter is true due to the fact that the best adversarial strategy is to fail only the links to the destination since in this case every such failure is reused once the same node appears again in the \emph{total prefix} of $w$ (see Claim \ref{claim:best_adv_strategy}).
Hence, we next compute the minimum number of distinct nodes $D$ in any set of $k$ random nodes. As we are interested in lower bounding $D$, we can choose $k$ minimal, i.e., $k=\frac{\varphi}{8\log n}$. The analysis follows from a \emph{balls-and-bins} argument where bins represent node IDs and balls are the $k$ positions that should be failed by the adversary. Thus, $D$ is a number of occupied bins (i.e., bins that contain at least one ball).
Let $D_i$ be a binary random variable indicating that the $i$-th ball falls into an empty bin (i.e., $D=\sum_{i=1}^k D_i$).
So, $
\Pr(D_i=1)\ge\frac{n-1-k}{n-1}$.
Since $k = \frac{\varphi}{8\log n}$ and $\varphi < n$, we obtain that:
\begin{align}
\Pr(D_i=1)\ge \frac{n-1-k}{n-1}&\ge \frac{8\log n -1}{8\log n}\ge 0.8.\nonumber
\end{align}
Thus,
$\E[D]=k\E[D_i]\ge 0.8k.
$
%
Now we can apply the Chernoff bound (for any $\delta \in (0,1]$):
\begin{align}
\Pr(D\le (1-\delta)0.8k)&\le \Pr(D\le (1-\delta)\E[D]) \nonumber\\
&\le e^{-\E[D]\delta^2/2} \le e^{-0.8k\delta^2/2} \nonumber.
\end{align}
By taking $\delta=0.5$ we obtain $
\Pr(D\le 0.4k)\le e^{-0.1k}.
$

It remains to prove that this bound still holds under the union bound for all $\binom{n-1}{\sqrt{\varphi}}$ possible sets of sequences that the adversary can choose. In other words, we have to ensure that $\binom{n}{\sqrt{\varphi}}e^{-0.1k} \le \frac{1}{n}$ (we took a larger number, since: $\binom{n}{\sqrt{\varphi}}\ge \binom{n-1}{\sqrt{\varphi}}$).
\begin{align}
\binom{n}{\sqrt{\varphi}}e^{-0.1k} &\le n^{\sqrt{\varphi}}e^{-0.1k}
= n^{\sqrt{\varphi}}e^{-\frac{\varphi}{80\log n}}\label{eq:union_reuse}\\
&= e^{\sqrt{\varphi}\ln n-\frac{\varphi}{80\log n}}
=e^{\varphi \left(\frac{\ln n}{\sqrt{\varphi}}-\frac{1}{80\log n}\right)} \nonumber\\
&\le e^{\varphi \left(\frac{\log n}{\sqrt{\varphi}}-\frac{1}{80\log n}\right)}. \nonumber
\end{align}

For $\varphi\ge 82^2\log^4 n$, we have $(\frac{\log n}{\sqrt{\varphi}}-\frac{1}{80\log n})\le \frac{-2}{82^2 \log n}$, and hence
$\binom{n}{\sqrt{\varphi}}e^{-0.1k} \le e^{\frac{-2\varphi}{82^2\log n}}
\le e^{-2\log^3 n} \le \frac{1}{n^2}$.
Since
$\binom{n}{\sqrt{\varphi}}\Pr(D\le 0.4k) = \binom{n}{\sqrt{\varphi}}\Pr (D\le \frac{0.4\varphi}{8\log n})
\le \frac{1}{n^2}$, w.h.p., any set of $\sqrt{\varphi}$ sequences (i.e., $w$-prefixes) will require $\Omega(\frac{\varphi}{\log n})$ failures.
\hfill $\Box$\end{proof}

\subsection{Deterministic Failover}\label{ssec:det}

Theoretically, the result of Theorem~\ref{thm:rand_strong_adv} can be derandomized, i.e., the $\rfs$ scheme can \emph{deterministically} ensure low loads. The idea is that we could \emph{verify} whether an (improbable) situation occurred and the random sequences generated by $\rfs$ actually yield a \emph{high} load (we just need to check all possible loads at any $w$); if so, another set of random permutations is generated. However, this verification is computationally expensive.

We hence now initiate the discussion of efficient deterministic schemes. In particular, we propose an optimal failover scheme (which matches our lower bound in Section~\ref{sec:worstcase}), at least for small $\varphi$.
Similar to $\rfs$, the deterministic failover scheme $\dfs$ is defined by a \emph{failover matrix} $\delta_{i,j}$; however, here $\delta_{i,j}$ will simply refer to a node's \emph{index} (and not the node itself): We define the index of any node $v_{\ell}$ to be $\ell-1$, i.e., the nodes $\{v_1,v_2,\ldots ,v_n\}$ are mapped to the indices $\{0,1,\ldots ,n-1\}$. Given a destination node $v_n$, $\dfs$ is defined by the following index matrix:
\begin{center}
\begin{align*}
1, 2, 4,8,\ldots, \left(0+2^{\left\lfloor \log n\right\rfloor}\right)\mod n \nonumber\\
2, 3, 5,9,\ldots, \left(1+2^{\left\lfloor \log n\right\rfloor}\right)\mod n \nonumber\\
3, 4, 6,10\ldots, \left(2+2^{\left\lfloor \log n\right\rfloor}\right)\mod n \nonumber\\
\ldots
\end{align*}
\end{center}

In general, the index in sequence $i\in[1,\ldots, n-1]$ at position $j\in[1,\ldots,\left\lfloor \log n\right\rfloor]$ is $\delta_{i,j} = (i-1)+2^{j-1} \mod n$. For example, if the link $(v_1,v_n)$ fails, $v_1$ will reroute via the node with index $1$, i.e., via $v_2$; and so on.
We can show the following result.
\begin{theorem}\label{thm:dfs}
The $\dfs$ scheme achieves a maximum load of $\loadmax=O(\sqrt{\varphi})$ in any scenario with $\varphi<\lfloor \log n \rfloor$ failures.
\end{theorem}

\begin{proof}
We will prove something even stronger: the adversary cannot choose link failures such that any \emph{node} $w$ forwards more than $\sqrt{\varphi}$ flows.
Clearly, an upper bound on the node load is an upper bound on the (incident) links: in the worst case, $w$ will forward all traffic to the same link.
To create a high load at some node $w$, the adversary needs to find failover sequences in the matrix $\delta_{i,j}$
where the node $w$ appears close to the \emph{beginning} of the sequence, and fail \emph{all} the links $(v_i,v_n)$, where $v_i$ is a node preceding $w$ in a sequence: i.e., the adversary fails the \emph{total prefix} of $w$. Note that failing the links to the destination is the best strategy for the adversary as failures are automatically reusable in other sequences (see Claim \ref{claim:best_adv_strategy}).

The following two claims will help us to show that the adversary wastes its entire failure budget in order to achieve a maximum load of $\sqrt{\varphi}$.
\begin{clm}
Every node index participates in only $\left\lfloor \log n\right\rfloor$ sequences.
\end{clm}
\begin{proof}
The $\dfs$ failover matrix is defined as $\delta_{i,j} = (i-1)+2^{j-1} \mod n$, where $i\in [1,\ldots,n-1]$ and $j\in [1,\ldots,\left\lfloor \log n\right\rfloor]$. From this construction, it follows that there are no index repetitions in the matrix columns. Since there are $\left\lfloor \log n\right\rfloor$ columns, the claim follows.
\hfill $\Box$\end{proof}

\begin{clm}
For any node index $\ell$, all $\ell$-prefixes (sets of indices preceding $\ell$ in the sequences) are disjoint.
\end{clm}
\begin{proof}
Let us define $m = i-1$ and $k = \ell-1$.
The index in sequence $m\in[0,\ldots, n-2]$ at position $k\in[0,\ldots,\left\lfloor \log n\right\rfloor -1]$ is $m+2^k \mod n$.
Consider a sequence $m'$ where the index $w$ appears at position $k'$ and a sequence $m''$ where the index $\ell$ appears at position $k''$.
Without loss of generality, assume that $k''>k'$.
Let $m'+2^{k^*} \mod n$ and $m''+2^{k^{**}} \mod n$ represent the indices in the prefixes of $\ell$ in sequences $m'$ and $m''$ accordingly. Assume by contradiction that these indices are the same. We have that
\begin{align}
m'+2^{k'} &= m''+2^{k''} \mod n \nonumber\\
m'+2^{k^*} &= m''+2^{k^{**}} \mod n\text{ (assumption)}\nonumber
\end{align}

\noindent and hence

\begin{align}
m'-m'' &= 2^{k''} - 2^{k'} + n\cdot C_1 \nonumber\\
m'-m'' &= 2^{k^{**}} - 2^{k^*} + n\cdot C_2\nonumber
\end{align}

Therefore

\begin{align}
2^{k^{**}}-2^{k''}+2^{k'}-2^{k^*} = n\cdot C_3 \nonumber
\end{align}
\noindent where $C_1, C_2$ and $C_3$ are some integer constants.

Notice that $\max(2^{k^{**}},2^{k''},2^{k'},2^{k^*} ) < n$, so the only possible values for $C_3$ are: $\{-1,0,1\}$.
Moreover, $(2^{k^{**}}-2^{k''})<0$, while $(2^{k'}-2^{k^*})>0$, and since the absolute value of these differences is bounded by $2^{\left\lfloor \log n\right\rfloor -1} \le 0.5n$, we can write:
\begin{align}
-0.5n <2^{k^{**}}-2^{k''}+2^{k'}-2^{k^*} < 0.5n.\nonumber
\end{align}
Thus, $0$ remains the only possible value for $C_3$.
The values $\{2^{k^{**}},2^{k''},2^{k'},2^{k^*}\}$ are distinct since there are no repetitions in the columns of the sequence matrix. Since $2^{k''}>2^{k^{**}}+2^{k'}+2^{k^*}$, due to a geometric series argument (the largest element is greater than the sum of all previous elements), we can state that
\begin{align}
2^{k^{**}}-2^{k''}+2^{k'}-2^{k^*} < 0.\nonumber
\end{align}
We conclude that there is no integer constant $C_3$ satisfying our assumption $m'+2^{k^*} = m''+2^{k^{**}} \mod n$ (i.e., there are two identical indices in the $\ell$-prefixes).
\hfill $\Box$\end{proof}

\begin{figure}[t]
\centering
\includegraphics[width=.5\columnwidth]{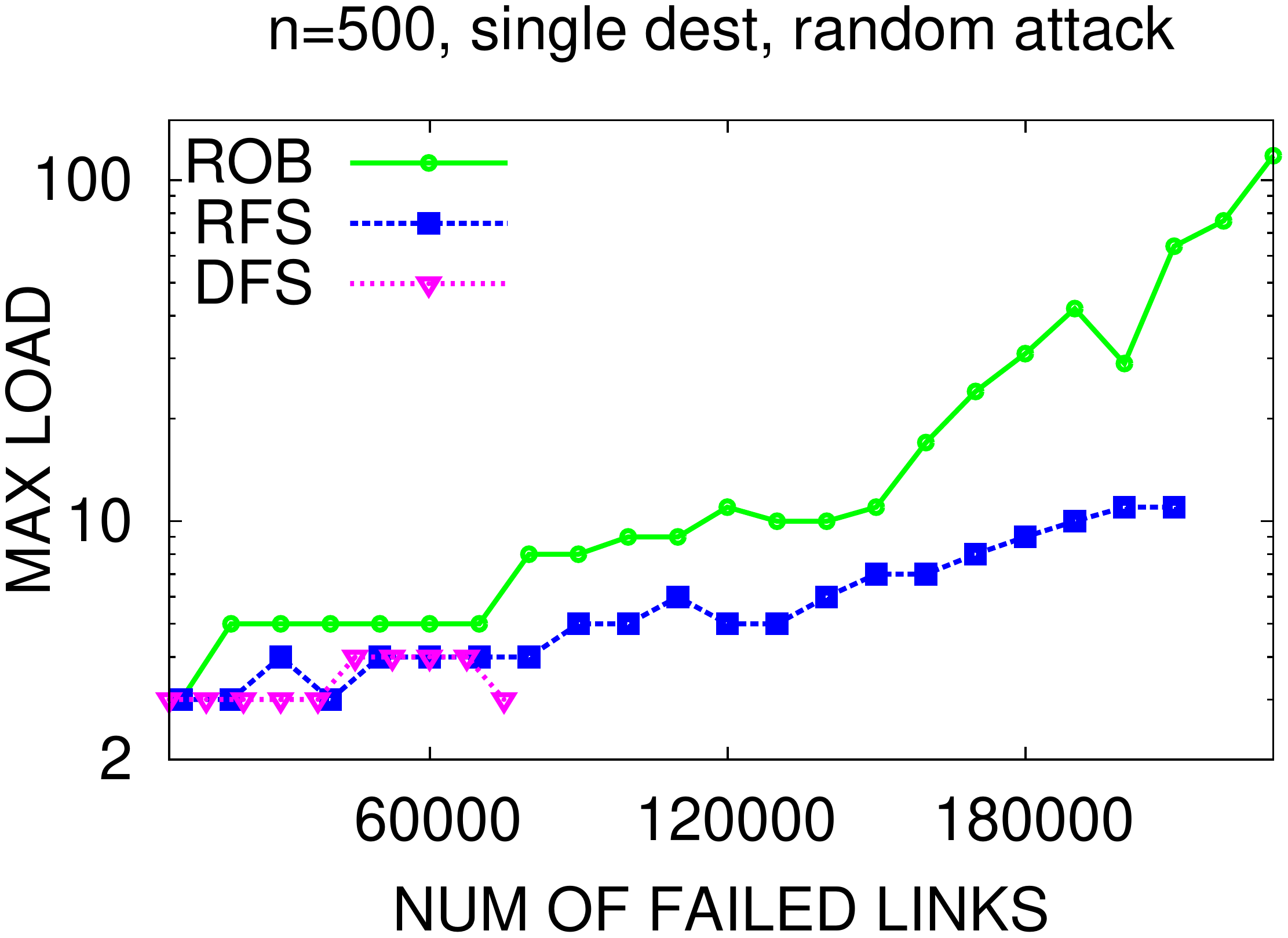}\\
\caption{Like Figure~\ref{fig:single-ecplise}, but under random failures.}\label{fig:ran}
\end{figure}

\begin{figure}[t]
\centering
\includegraphics[width=.49\columnwidth, height = 0.4\columnwidth]{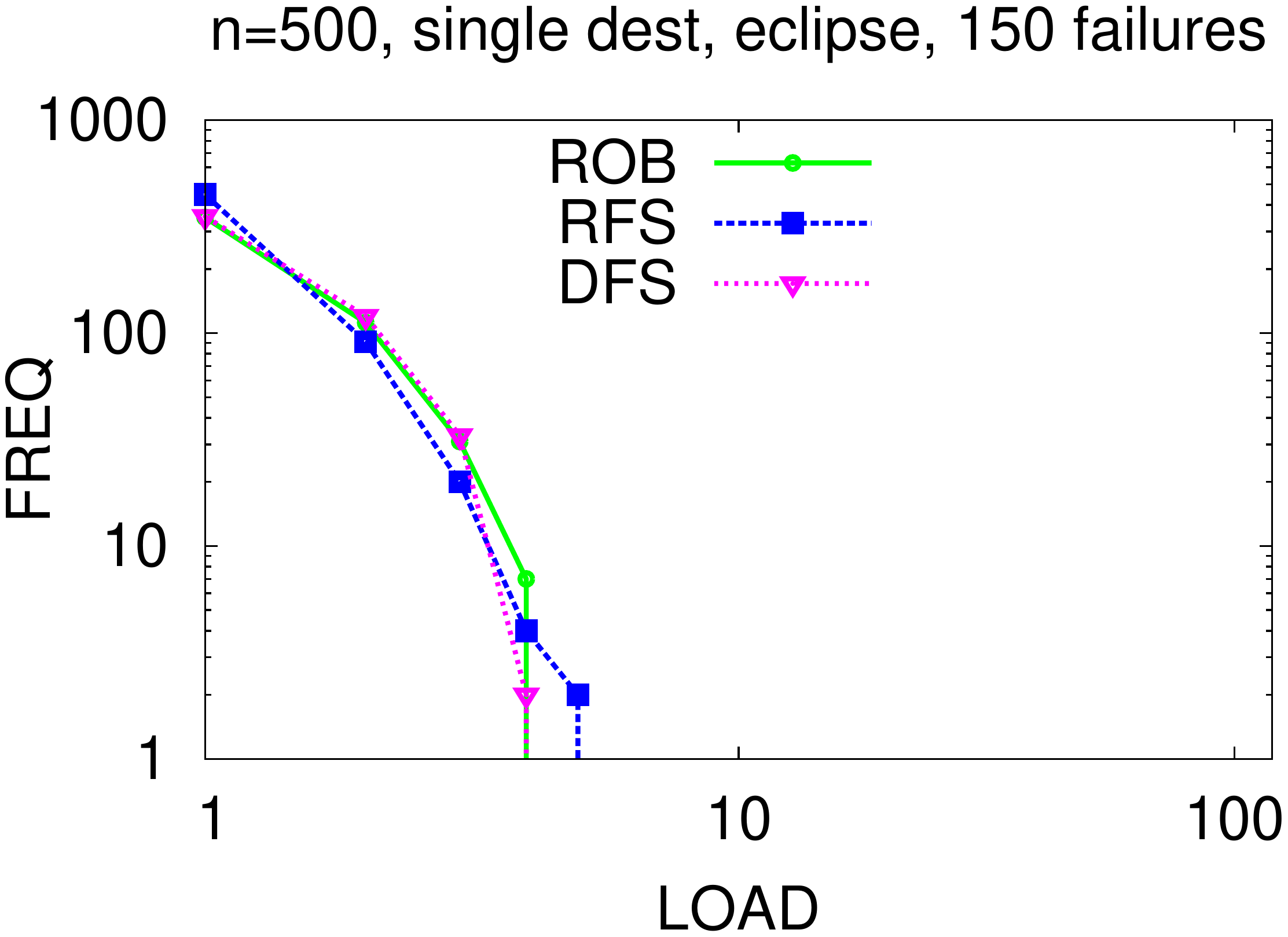}
\includegraphics[width=.49\columnwidth, height = 0.4\columnwidth]{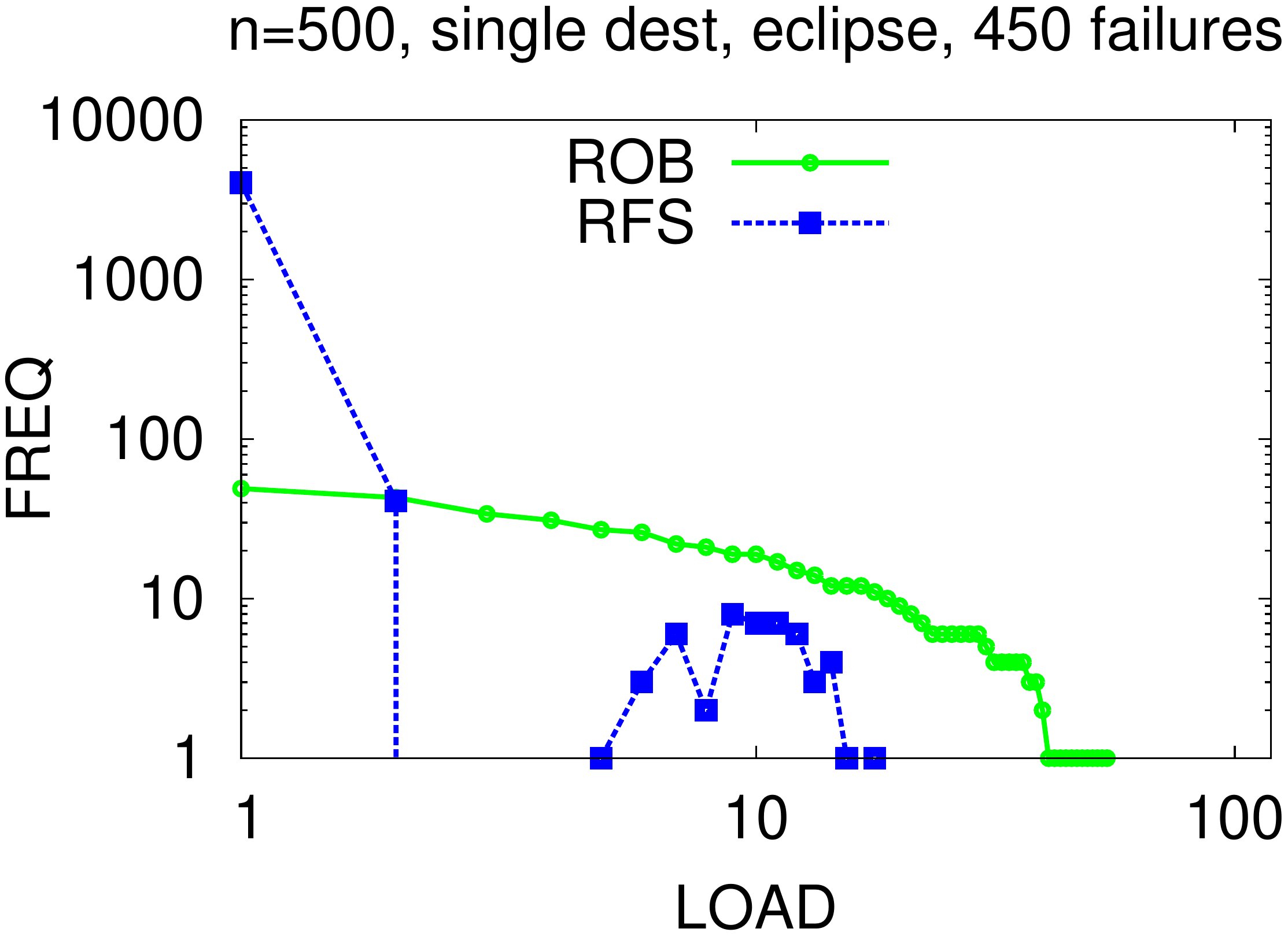}\\
\caption{Load distribution over links.}\label{fig:single-loaddist}
\end{figure}

Armed with these claims, we are ready to continue with the proof. Since all prefixes are disjoint, the adversary cannot reuse failures of one flow for another. Thus, the adversary will be able to route \emph{one} flow to $w$ using a single failure (by finding a sequence in which $w$ appears at the first position); to add another flow, the adversary takes a sequence $\delta_i$ in which $w$ is located at position 2 and will fail the links $(v_i, v_n)$, and $(v_{(\delta_{i,1})+1},v_n)$.  And so on.
Thus, the number of used failures can be represented as
\begin{align}
1+2+3+4+\cdots + L \le \varphi \nonumber
\end{align}
\noindent where $L$ is the number of flows passing through $w$ on the way to the destination $v_n$.
So:
\begin{align}
1+2+3+4+\cdots + L &\le \varphi \nonumber\\
\frac{L(L+1)}{2} & \le \varphi \nonumber\\
L<\sqrt{2\varphi}.\nonumber
\end{align}

\begin{figure}[t]
\centering
\includegraphics[width=.49\columnwidth, height = 0.4\columnwidth]{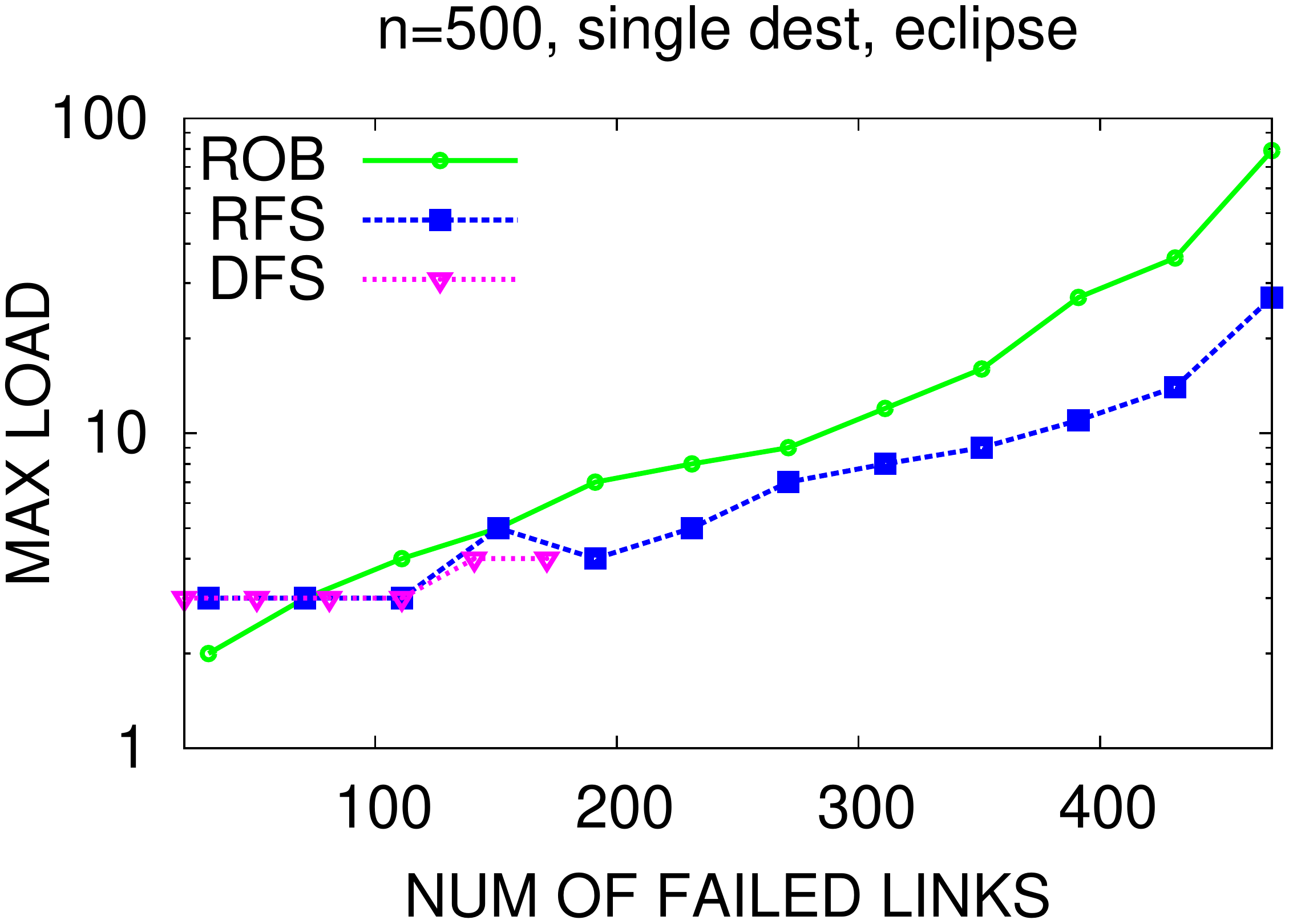}
\includegraphics[width=.49\columnwidth, height = 0.4\columnwidth]{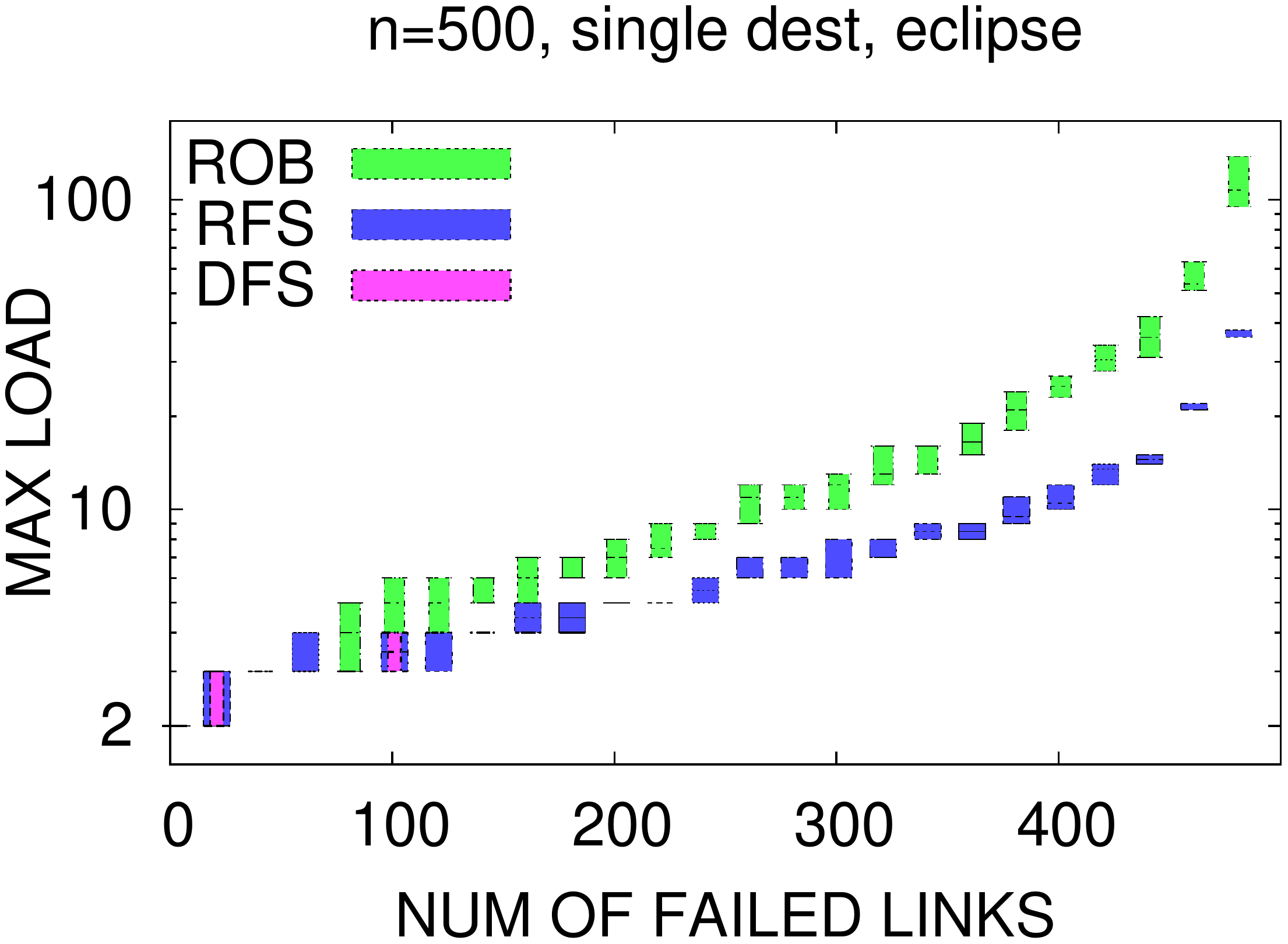}\\
\caption{\emph{Left:} Load
 for different algorithms ($n=500$).
\emph{Right:} Boxplot.}\label{fig:single-ecplise}
\end{figure}

Note that the index of the destination node ($n-1$ in our case) can appear inside the failover sequences. In this case, the index will be skipped since the link to it from the source already failed. By skipping one index, we shorten the failover sequence by 1, and since every sequence has length $\left\lfloor \log n\right\rfloor$, our failover scheme holds for any $\varphi < \left\lfloor \log n\right\rfloor$.
\hfill $\Box$\end{proof}

\newpage
\section{Beyond Worst-Case Failures}\label{sec:sims}

\begin{wrapfigure}{r}{0.5\textwidth}
\centering
\includegraphics[width=.5\textwidth]{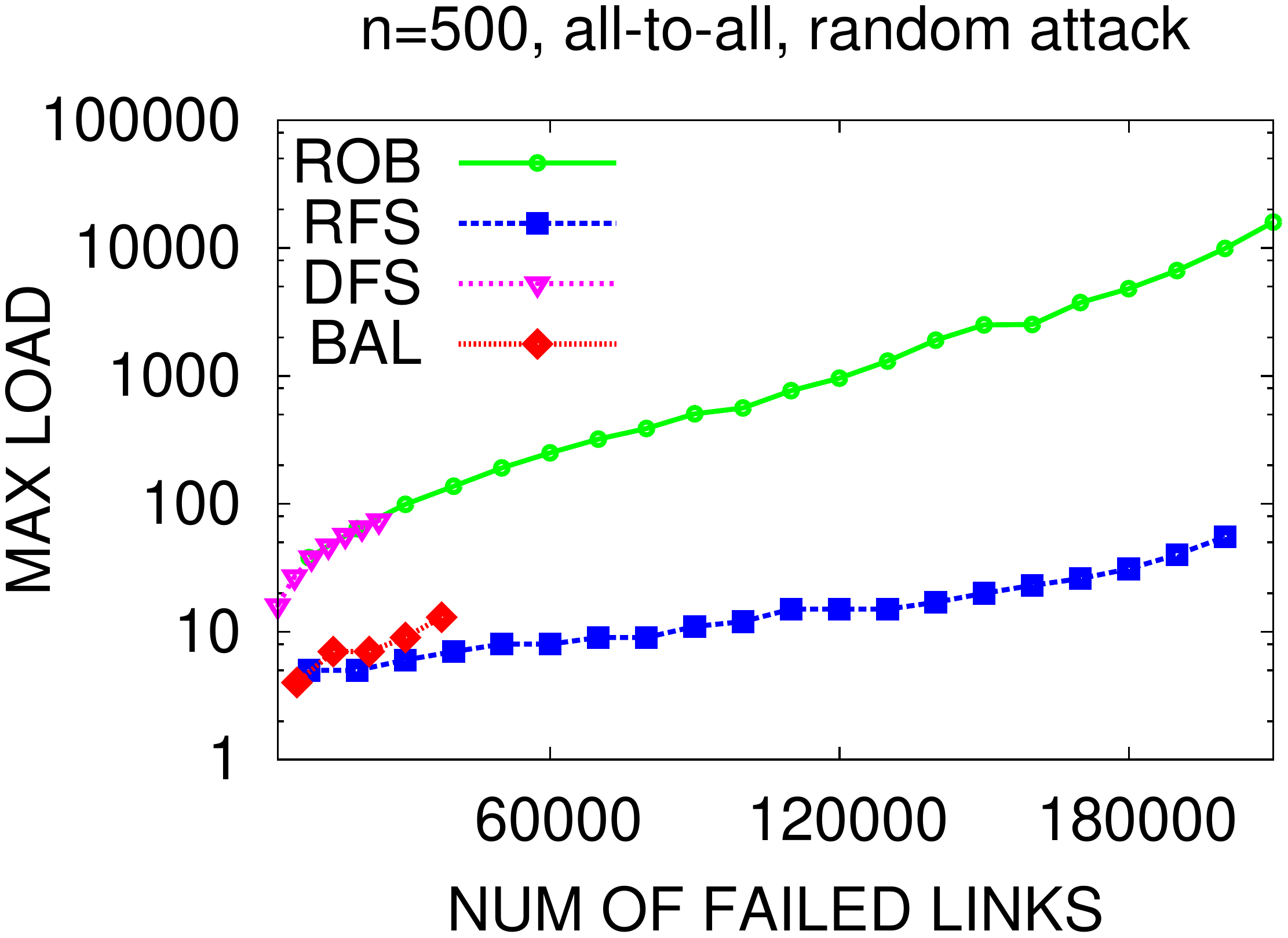}
\caption{Max load in all-to-all communication.}\label{fig:all-to-all-LOGy}
\end{wrapfigure}
To complement our worst-case bounds, we conducted simulations with different failure scenarios. (1) \textbf{Ran}: links are failed uniformly at random; (2) \textbf{Ecl} (an ``eclipse attack''): links are removed at random around destination $v_n$. We used two traffic patterns. (1) \emph{Single dest:} one unit of flow from each node to $v_n$; (2) \emph{all-to-all:} one unit of flow from each node to every other node. In addition to our failover schemes $\rfs$ and $\dfs$, we also simulate the following naive strategies. (1) \textbf{Bal} (``balanced''): If the destination cannot be reached directly, forward to an available port chosen uniformly at random (depending on the destination and the set of failed links). This strategy seeks to balance traffic but does not ensure loop-freeness. (2) \textbf{Rob} (``robust''): If the destination cannot be reached, forward to the available neighboring switch which has the lowest identifier (in a modulo manner, and assuming that switches have unique identifiers).
We start with the \emph{single dest} traffic pattern.
Figure~\ref{fig:single-ecplise} (\emph{left}) plots the load as a function of the number of failed links under \textbf{Ecl}. (Note the logarithmic scale on the y-axis.)
We observe that compared to Theorem~\ref{thm:rand_strong_adv} which deals with worst-case failures, $\rfs$ performs significantly better: while our conservative bound suggests that to create a load of $10=\sqrt{\varphi}$,
$\varphi/\log{n}=100/\log{500}=11.15$ failures are needed, more than 300 are necessary in our experiment. Moreover, we observe that $\rfs$ yields a much lower load than the naive approach \textbf{Rob}; for the single-destination scenario, \textbf{Bal} is similar to \textbf{Rob} and is not shown explicitly here.
The $\dfs$ algorithm also gives a low load; however, it is only defined up to a certain $\varphi$ (see Theorem~\ref{thm:dfs}). The variance of these experiments is typically small, see the boxplot in Figure~\ref{fig:single-ecplise} (\emph{right}).

As expected, under \textbf{Ran} failures, the load is generally lower, and our scheme can tolerate more failures without creating loops  (see Figure~\ref{fig:ran}).

As the maximal link load reveals a partial picture only, Figure~\ref{fig:single-loaddist}
studies the load distribution over multiple links (under \textbf{Ecl}), once for 150 failures (\emph{left}) and once for 450 failures
(\emph{right}). Obviously, most links hardly contain more than one or two flows under $\rfs$;
again, under the naive \textbf{Rob} strategy (and similarly for \textbf{Bal}), the situation is worse.

Let us have a look at alternative traffic matrices. Figure~\ref{fig:all-to-all-LOGy}
shows the results for an all-to-all communication pattern (under \textbf{Ran}). Interestingly,
for $\rfs$,
the load is not much higher than in the single-destination scenario; this confirms the
good load-balancing properties of $\rfs$. However, we also see
that $\dfs$ performs poorly and needs to be generalized for the multi-destination scenario.
Finally, we note that in this scenario, we can exploit \textbf{Bal}'s flexibility and in contrast
to the single destination case, the algorithm significantly outperforms \textbf{Rob} (in terms of load).




\section{Related Work}\label{sec:relwork}

This work is motivated by the trend towards Software-Defined Networking and in particularly the fast failover mechanism which supports the in-band masking of failures (see Section~5.8 of the OpenFlow 1.1 specification). However, as the convergence time of routing algorithms is often relatively high compared to packet forwarding speeds, ranging from 10s of milliseconds to seconds depending on the network~\cite{podc12shenker}, many networks today incorporate some robustness already in the forwarding tables of a router or switch: Thus, robust routing concepts and link protection schemes have been studied intensively for many years, also outside SDN.

For example, robust Multiprotocol Label Switching (MPLS) supports local and global path protection to compute shortest backup paths around an outage area~\cite{mpls-fast-reroute,r37}, where ``shortest'' is often meant in terms of congestion~\cite{r32,r38}. Related to our connectivity and load-balancing tradeoff is also the work by Suchara et al.~\cite{jrex-fail} who analyze how to jointly optimize traffic engineering and failure recovery from pre-installed MPLS backup paths. However, in contrast to our paper, their solution is path-based and not local, and the focus is on robust optimization.

Alternative solutions to make routing more resilient rely on special header bits (e.g., to determine when to switch from primary to backup
paths, as in \emph{MPLS Fast Reroute}~\cite{mpls-fast-reroute}, or to encode failure information to make
failure-aware forwarding decisions~\cite{header1,header2}), or on fly table modifications~\cite{on-fly}.
Recently, Feigenbaum et al.~\cite{podc12shenker} made an interesting first step towards a better theoretical understanding of resilient SDN tables.
The authors prove that routing tables can provide guaranteed resilience (i.e., loop-freeness) against a \emph{single} failure, when the network remains connected.



\section{Conclusion}\label{sec:conclusion}


So, will or won't you shoot in your foot with fast failover? Our results show that there exists an interesting tradeoff between a ``safe'' and ``efficient'' failover. The usefulness of the local failover depends on whether link failures are rather adversarial or random, and on how flexibly the failover rules can be specified.
In particular, we have seen that the possibilities of destination-based failover schemes are very limited. But also more expressive failover schemes where flows can be forwarded depending on arbitrary local matching rules (this is more general than today's OpenFlow specification), can lead to high network loads in the worst case.
On the positive side, relatively simple algorithms exist which match these lower bounds.


\textbf{Acknowledgments.}
We would like to thank Chen Avin for valuable discussions and advice.

  \bibliographystyle{abbrv} 
  



\newpage
\begin{appendix}

\section{Algorithms Bal and Rob}

The $\Bal$ algorithm is executed sequentially for each node $v_i$ and for each failed link incident to $v_i$  chooses the new next hop based on the
current node ($v_i$) and the second end of the currently failed link.

\renewcommand{\algorithmicrequire}{\textbf{Input:}}
\renewcommand{\algorithmicensure}{\textbf{Output:}}

\begin{algorithm}[h]
    \caption{\textbf{$\Bal$}}
    \label{alg:bal}
    \begin{algorithmic}[1]
    \REQUIRE current node -- $v_i$, second end of failed link -- $v_j$
    \ENSURE new next hop
    \IF {$i>j$}
    \STATE $next\_hop = (i+j+1) \mod n$
    \ELSE
    \STATE $next\_hop = (i-j+1) \mod n$
    \ENDIF
    \WHILE{$(v_i,v_{next\_hop})$ is failed}
    \STATE $next\_hop = (next\_hop+1) \mod n$
    \ENDWHILE
    \RETURN $next\_hop$
    \end{algorithmic}
\end{algorithm}

The $\Rob$ algorithm always selects the next available neighbor as the new next hop.

\begin{algorithm}[h]
    \caption{\textbf{$\Rob$}}
    \label{alg:rob}
    \begin{algorithmic}[1]
    \REQUIRE current node -- $v_i$
    \ENSURE new next hop
    \STATE $next\_hop = (i+1) \mod n$
    \WHILE{$(v_i,v_{next\_hop})$ is failed}
    \STATE $next\_hop = (next\_hop+1) \mod n$
    \ENDWHILE
    \RETURN $next\_hop$		
    \end{algorithmic}
\end{algorithm}

\section{Discussion}

In this section we take a broader look at the problem and study additional scenarios and adversarial models.

\subsection{Realistic Adversary}

So far, we have concentrated on a worst-case adversary which fails exactly those links which will lead to the highest load. We now want to have a look at a more realistic, ``weak adversary'' $\textbf{Ran}$, who fails links at random (\emph{independently} of the installed failover rules).

\begin{theorem}\label{thm:rand_weak_adv}
Using the $\rfs$ failover scheme, in order to create a load of $\sqrt{\varphi}$, $\textbf{Ran}$ will have to fail at least $\Omega \left(\varphi\right)$ links \emph{w.h.p.} The same results hold even if $\textbf{Ran}$ knows the destination of the traffic.
\end{theorem}
\begin{proof}
We build upon the proof of Theorem~\ref{thm:rand_strong_adv}. Let us assume that $\textbf{Ran}$ knows the destination node $v_n$. Then, it will only fail the links incident to the destination: Consider any failover sequence $\delta_i$ and assume that the $j$ first nodes in this sequence are already in use as backup nodes. In order to make use of the node $j+1$, the adversary needs to fail either the link $(\delta_{i,j},v_n)$ or the link $(\delta_{i,j-1},\delta_{i,j})$. If the adversary will fail only links to the destination $v_n$, the probability of reaching the node $\delta_{i,j+1}$ is at least $1/n$, while a random link selection will succeed only with probability $2/(n(n-1))$. Moreover, such failures (i.e., failures of links towards the destination) can be reused better (see Claim \ref{claim:best_adv_strategy}).

From the proof of Theorem \ref{thm:rand_strong_adv}, we know that the minimal \emph{total prefix length} (i.e., the sum of all $\sqrt{\varphi}$-prefixes) of any node $w$ is $\frac{\varphi}{8\log n}$ and the number of distinct nodes in the prefix is at least $d=\frac{0.4\varphi}{8\log n}=\frac{0.05 \varphi}{\log n}\ge \frac{0.03 \varphi}{\ln n}$. The number of possible total prefixes is $\binom{n-1}{\sqrt{\varphi}}$. We now ask the following question: if the adversary fails $r$ (of course $r > d$) random links adjacent to the destination, what is the probability that these failures will cover at least one possible \emph{total} prefix? First, we will compute the probability $p_1$ that $r$ random failures will cover a single total prefix (with $d$ distinct nodes); subsequently, we will apply a union bound argument over all possible total prefixes. Simple combinatorics give us the expression for $p_1$:
\begin{align}
p_1 = \frac{\binom{n-d}{r-d}}{\binom{n}{r}} = \frac{(n-d)!r!}{(r-d)!n!}. \nonumber
\end{align}
Using Stirling's approximation for the factorial, we have:
\begin{align}
\sqrt{2\pi}\cdot a^{0.5+a} e^{-a} \le a!\le e\cdot a^{0.5+a} e^{-a}. \nonumber
\end{align}

Therefore,
\begin{footnotesize}
\begin{align}
p_1 &\le \frac{e(n-d)^{0.5+(n-d)}e^{-(n-d)} e\cdot r^{0.5+r} e^{-r}}{\sqrt{2\pi}(r-d)^{0.5+(r-d)}e^{-(r-d)} \sqrt{2\pi}\cdot n^{0.5+n} e^{-n}} \nonumber\\
&= \frac{e^2(n-d)^{0.5+(n-d)} r^{0.5+r}}{2\pi(r-d)^{0.5+(r-d)} n^{0.5+n}} \nonumber\\
&\le \frac{1.2 n^{0.5+(n-d)} r^{0.5+r}}{(r-d)^{0.5+(r-d)} n^{0.5+n}} \nonumber\\
&= 1.2 n^{-d} r^{0.5+r}(r-d)^{-(0.5+(r-d))} \nonumber\\
&= \left(1.2(r-d)^{-0.5}\right) n^{-d} r^{0.5+r}(r-d)^{-(r-d)} \nonumber\\
&\le n^{-d} r^{0.5+r}(r-d)^{-(r-d)}\label{eq:rAndD}\\
&= e^{-d\ln n +(0.5+r)\ln r -(r-d)\ln{(r-d)}}.\nonumber
\end{align}
\end{footnotesize}
The Equation~(\ref{eq:rAndD}) is true since $r>d$, and assuming that $r-d\ge 2$, we get: $1.2(r-d)^{-0.5}\le 1$.
Now recall that $d=\frac{0.03\varphi}{\ln n}$, and take $r=d\ln n=0.03\varphi$. Then:
\begin{footnotesize}
\begin{align}
p_1 &\le e^{-0.03\varphi +(0.5+0.03\varphi)\ln 0.03\varphi -\frac{0.03\varphi(\ln{n} -1)}{\ln{n}}\ln{\frac{0.03\varphi(\ln{n} -1)}{\ln{n}}}} \nonumber \\
&= e^{0.03\varphi\left( -1 + (\frac{0.5}{0.03\varphi}+1)\ln 0.03\varphi -(1-\frac{1}{\ln n})\ln{\frac{0.03\varphi(\ln{n} -1)}{\ln{n}}} \right)} \nonumber \\
&= e^{0.03\varphi\left( -1 + \frac{0.5\ln 0.03\varphi}{0.03\varphi}+\frac{\ln 0.03\varphi}{\ln n} -(1-\frac{1}{\ln n})\ln{\frac{\ln{n} -1}{\ln{n}}} \right)} \nonumber \\
&= e^{0.03\varphi\left( -1 + \frac{\ln 0.03\varphi}{0.1\varphi}+\frac{\ln 0.03\varphi}{\ln n} +(1-\frac{1}{\ln n})\ln{\frac{\ln{n}}{\ln{n}-1}} \right)}. \nonumber
\end{align}
\end{footnotesize}

Using the following Taylor expansion (for $a>1$):
\begin{align}
\ln{\frac{a}{a-1}}&=-\ln{\left(1-\frac{1}{a}\right)} \nonumber\\
&=\sum_{i=1}^{\infty}\frac{1}{ia^i} \nonumber\\
&=\frac{1}{a}+\frac{1}{2a^2}+\sum_{i=3}^{\infty}\frac{1}{ia^i} \nonumber\\
&\le \frac{1}{a}+\frac{1}{2a^2}+\frac{1}{2a^2} \nonumber\\
&\le \frac{1}{a}+\frac{1}{a^2}, \nonumber
\end{align}
\noindent we obtain
\begin{footnotesize}
\begin{align}
p_1 &\le e^{0.03\varphi\left( -1 + \frac{\ln 0.03\varphi}{0.1\varphi}+\frac{\ln 0.03\varphi}{\ln n} +(1-\frac{1}{\ln n}) \left(\frac{1}{\ln{n}}+\frac{1}{\ln^2{n}}\right) \right)} \nonumber \\
&= e^{\frac{0.03\varphi}{\ln n}\left( -\ln n + \frac{\ln 0.03\varphi \ln n}{0.1\varphi}+\ln 0.03\varphi +(\ln n - 1) \left(\frac{1}{\ln{n}}+\frac{1}{\ln^2{n}}\right) \right)} \nonumber \\
&= e^{\frac{0.03\varphi}{\ln n}\left( -\ln n + \frac{\ln 0.03\varphi \ln n}{0.1\varphi}+\ln 0.03\varphi + 1 + \frac{1}{\ln{n}}-\frac{1}{\ln{n}}-\frac{1}{\ln^2{n}} \right)}. \nonumber
\end{align}
\end{footnotesize}

Note that $\varphi\le n$ and thus for $\varphi\ge 10 \ln^4 n$,
$$\frac{\ln 0.03\varphi \ln n}{0.1\varphi} \le \frac{1}{\ln^2{n}}.$$
We can write:
\begin{align}
p_1 &\le e^{\frac{0.03\varphi}{\ln n}\left( -\ln n +\ln 0.03\varphi + 1  \right)} \nonumber\\
&\le e^{\frac{0.03\varphi}{\ln n}\left( -\ln n +\ln 0.03n + 1  \right)}\nonumber\\
&\le e^{\frac{0.03\varphi}{\ln n}\left(-2.5  \right)}\nonumber\\
&= e^{\frac{-0.075\varphi}{\ln n}}. \nonumber
\end{align}

Since $e^{-\frac{0.075\varphi}{\ln n}}\le e^{-\frac{\varphi}{80\log n}}$, we can use the union bound from the proof of Theorem \ref{thm:rand_strong_adv} (see Equation (\ref{eq:union_reuse}) in that proof): we conclude that for $\varphi\ge 82^2\log^4 n$, in order to create a load of $\sqrt{\varphi}$, a weak adversary will have to fail at least $r=d\ln n = 0.03 \varphi$ links with probability of at least $1-\frac{1}{n^2}$.
\hfill $\Box$\end{proof}

\subsection{Multiple Destinations}

We can also generalize the traffic model: instead of having a single destination, each node $v_i$ for all $i\in [1,\ldots,n]$ needs to send one unit of flow to each other node $v_j$ with $j\neq i$. We can adapt the $\rfs$ scheme in the sense that we generate $\binom{n}{2}$ random failover permutations, one for each source-destination pair. Note that this corresponds to a matrix representation $\delta_{\cdot,\cdot}$ with $\binom{n}{2}$ rows and $n-2$ columns. Again, the permutations ensure loop-freeness.

\begin{theorem}\label{thm:mult_dest}
Using the (generalized) $\rfs$ failover scheme, in order to load a node with $\varphi$ flows ($0<\varphi < n$), the adversary will have to fail $\Theta(\varphi)$ links.
\end{theorem}
\begin{proof}
In the generalized $\rfs$, we have $n(n-1)$ failover sequences: each set of $n-1$ sequences are for a specific destination. Consider now the first column of this $n(n-1)\times (n-2)$ matrix. Since there are $n(n-1)$ rows, there exists a node $w$ that appears at least $n-1$ times in the first column. Now, for each row $i$ in which $\delta_{i,1}=w$, the adversary can fail the link from $v_i$ to the appropriate destination, thus making the flow to go through $w$. Hence, a single failure is enough to bring an additional flow to $w$, i.e., the adversary needs at most $\varphi$ failures in order to create a load of $\varphi$ on $w$.

It remains to show that these failures cannot be reused by the adversary, i.e., the adversary indeed will need to invest $\varphi$ failures.
The argument is simple. In order to create a load of $\varphi$ on $w$, the adversary has to `finish' the prefixes of $\varphi$ sequences. In order to do so, it first needs to fail the direct $(v_1,v_n)$ link of every such sequence, i.e., for each sequence at least this link has to be failed. But these links are unique (each sequence serves a different source-destination pair) and thus cannot be reused. The last implies that the adversary indeed needs to invest at least $\varphi$ failures.
\hfill $\Box$\end{proof}


\end{appendix}

%

\end{document}